\documentclass[11pt,english]{article}
\usepackage[T1]{fontenc}
\usepackage{color}
\usepackage{amsmath}
\usepackage{amssymb}
\usepackage{setspace}
\usepackage{esint}
\usepackage{pdfsync}
\makeatletter

\usepackage{babel}
\usepackage{amsthm}
\usepackage{graphicx}
\usepackage{amsfonts}
\usepackage{slashed}
\usepackage{afterpage}
\usepackage{cite}
\usepackage{bbold}

\definecolor{verde}{cmyk}{.83,.21,1,.08}

\newcommand{\cinf}{{C^\infty({\cal M})}}
\newcommand{\M}{\mathcal M}
\newcommand{\A}{\mathcal A}
\newcommand{\HH}{\mathcal H}
\newcommand{\ds}{\slashed \partial}

\font\mybb=msbm10 at 12pt
\def\bb#1{\hbox{\mybb#1}}

\def\be{\begin{equation}}
\def\ee{\end{equation}}
\def\bea{\begin{eqnarray}}
\def\eea{\end{eqnarray}}

\newcommand{\del}{\partial}

\newcommand{\I}{\mathbb I}

\newtheorem{prop}{Proposition}[section]

\newtheorem{df}[prop]{Definition}
\newtheorem{rem}[prop]{Remark}

\numberwithin{equation}{section}

\usepackage{babel}
\makeatother

\usepackage{babel}
\begin{document}

\title{Lorentz signature and twisted spectral triples}
\author{A. Devastato$^a$, S. Farnsworth$^b$, F. Lizzi$^{acd}$ and P. Martinetti$^e$}
\date{}
\maketitle
 \begin{center}
    $^a$  INFN sezione di Napoli,\\
$^b$ Max Planck Institute for Gravitational Physics (Albert Einstein Institute), Germany,\\ 
$^c$ Dipartimento di Fisica ``E. Pancini'', Universit\`a di Napoli {\sl Federico II},\\
$^d$  Institut de C\'\i encies del Cosmos (ICCUB),
Universitat de Barcelona,\\
$^e$ Dipartimento di Matematica, Universit\`a di Genova.
  \end{center} 
\footnotetext{\hspace{0truecm}  astinodevastato@gmail.com, shane.farnsworth@aei.mpg.de,
  fedele.lizzi@na.infn.it,  martinetti@dima.unige.it}
\begin{abstract}
We show how twisting the spectral triple of the Standard Model of
elementary particles naturally yields the Krein space associated with
the Lorentzian signature of spacetime. We discuss the associated
spectral action, both for fermions and bosons.
What emerges is a tight link between twists and Wick rotation.
\end{abstract}

\section{Introduction}

Noncommutative differential geometry (NCG) provides a unified
framework from which to describe both Einstein-Hilbert gravity (in
Euclidean signature) and classical gauge theories~\cite{Connes:1994kx}. In particular, it gives an elegant description of the full Standard Model of particle physics in all of its detail, including the Higgs mechanism and neutrino
mixing, as gravity on a certain ``almost commutative manifold''~\cite{Chamseddine:1996zu, Chamseddine:2007oz}. A recent and comprehensive review can be found in~\cite{Walterbook}.

The main benefit of the NCG approach to physics is that it offers a more constrained description of gauge theories than the usual effective field theory approach. Indeed, the added geometric constraints impose a range of successful and phenomenologically accurate restrictions on the allowed particle content of the Standard Model of particle physics~\cite{Chamseddine:2008uq, Boyle:2014wba, Brouder:2015qoa,Devastato:2013wza}. Despite this success, an early estimate for the Higgs mass was also furnished at $m_{H}\simeq 170~\text{Gev}$. This prediction was disfavored by the Tevatron data, and has since been ruled out by the LHC~\cite{Aad,chat2012}. While falling short of an accurate comparison with experiment, this prediction depended on a number of assumptions including the big desert hypothesis, as well as the presence of a scale at which the coupling constants of the three gauge interactions unify. 

In light of the many successes of the NCG construction, one is led to
question the various assumptions that went into the incorrect Higgs
mass calculation, and above all the validity of the big desert
assumption. This concern is particularly pressing as the ``low''
experimental value detected for the Higgs mass causes an instability
(or meta-stability) in the electroweak vacuum at intermediate energies
(see \cite{near-critic} for a recent update), that may be cured by the
addition of a new scalar field suitably coupled to the Higgs - usually
denoted $\sigma$ (e.g. \cite{C.-S.-Chen:2012vn,Elias-Miro:2012ys}). If
the addition of such a scalar field were admissible within the NCG
construction, it would not only stabilize the electroweak vacuum, but
also allow compatibility with the experimentally observed Higgs mass \cite{Chamseddine:2012fk}. Such an outcome is difficult to achieve however as the extra geometric constraints imposed by the NCG formalism severely restrict any allowed addition to the Standard Model. Early attempts in~\cite{Stephan:2009fk,Stephan:2013fk} to generate an extra scalar field within the framework required the adjunction of new fermions. 
More recently, there have been a number of  phenomenologically
viable Standard Model extensions singled out by the NCG
framework which preserve the fermionic sector. This has been achieved
for example in ~\cite{Chamseddine:2013fk,Chamseddine:2013uq}  by
relaxing some of the geometrical constraints, and
in~\cite{Farnsworth:2014vva} by taking full account of the outer
symmetries of the model (let us also mention some proposals to modify
the grading, based on Morita equivalence, developed in \cite{DAndrea:2015aa,Dabrowski:2017aa}, as well as other
modification of the grading in \cite{T.-Brzezinski:2016aa}).
\medskip

In this paper we are mostly concerned with the outcome of an extension of the Standard Model known as \emph{grand symmetry}, proposed
in~\cite{Devastato:2013fk} (see \cite{Devastatao:2014xga} for a
shorter non-technical presentation). It relies on an
enlargement of the Standard Model algebra, and ultimately allows
one to obtain the field $\sigma$ in agreement with the NCG principles, namely 
as an internal part of a connection.
As shown in~\cite{buckley} however, in order to make the grand symmetry extension work, one is required to twist the noncommutative geometry of the Standard Model, in the sense of Connes-Moscovici\cite{Connes:1938fk}. Besides solving some
technical difficulties, the twist also permits one to understand the
breaking of the grand symmetry down to the Standard Model symmetries
in a dynamical process induced by the spectral action.  Our first result establishes the fact that \emph{the required twist  corresponds to a Wick
  rotation.} More precisely, we show in section \ref{sec:twistlorentz}
that the twist turns the inner product of the
Hilbert space of (Euclidean) spinors into a Krein product. The latter
is precisely the inner product associated with spinors on a Lorentzian manifold. In a sense made
precise in \S \ref{subsec:lorentzkrein}, the twist is actually the square of the Wick rotation. 

Our second result concerns the spectral action in the twisted context.
 While the behavior of gauge transformations for twisted spectral
triples has already been
worked out in \cite{Landi:2017aa}, the corresponding gauge invariance of the
spectral action has not yet been addressed.  We  investigate this
question in section
\ref{sec:action}:
\begin{itemize}
\item We begin with the fermionic action $S^F$ in \S
  \ref{subsec:fermact}, showing that the straightforward adaptation to
  the twisted case of the usual formula is indeed invariant under a
  twisted gauge transformation. However, 
  it is not antisymmetric when restricted to the (positive)
  eigenvectors of the chirality operator, unlike the non-twisted case.  This leads us to propose two
  possible definitions of $S^F$ in the twisted context: either by
  restricting to the eigenvectors of the unitary implementing the
  twist, or by considering a Dirac operator, which is Krein-hermitian rather than
  hermitian.
\item The bosonic action $S^B$  is adressed in \S \ref{subsec:bosonact}. We
  show that there is an easy way to rewrite it so that it becomes
  invariant under a twisted gauge transformation. We also investigate 
  the possibility of a Krein adjoint Dirac operator. In both cases, our formulas
  give back the Euclidean bosonic action. 
\end{itemize}

\noindent The picture that emerges is that  
twisted geometries may provide 
an appropriate framework from which to facilitate the
description of non-Euclidean signatures in NCG. 

We begin in the following section by recalling the main
feature of the twisted spectral triple of the Standard Model.
\smallskip

\section{Twisted spectral geometry for the standard model}
\label{sec:SM}

This section deals with the twisted spectral triple
of the Standard Model of \cite{buckley}. We do not discuss the usual non-twisted version, which can be
found in \cite{Chamseddine:2007oz}; neither do we motivate the
importance of twists in
noncommutative geometry. Let us just recall that twisted spectral triples were introduced in~\cite{Connes:1938fk} in order to build spectral triples from 
algebras which do not exhibit a trace. Quite unexpectedly, they also provide the correct mathematical framework to write the ``beyond
SM'' Grand symmetry model of  \cite{Devastato:2013fk}.

A \emph{twisted spectral triple}
$\left(\mathcal{A},\mathcal{H},D; \rho\right)$
consists of a {*}-algebra $\mathcal{A}$ of bounded operators in
a Hilbert space{\footnote{We denote $T^\dag$ the adjoint of an
    operator $T$ on $\HH$. As usual, we omit the symbol of
    representation  for the algebra and identify
    $\pi(a^*)=\pi(a)^\dag$ with $a^*$. }} $\mathcal{H}$, together with a non-necessarily bounded self-adjoint
operator $D$ on $\mathcal{H}$ with compact resolvent, and an
automorphism $\rho$ of $\A$ such that the twisted commutator 
\begin{equation}
[D,a]_{\rho}:=Da-\rho(a)D\label{eq:1}
\end{equation}
is bounded for any $a$ in $\mathcal{A}$. 
The twisted spectral triple is \emph{even} if there is a 
$\mathbb{Z}{}_{2}$ grading, i.e. an operator $\Gamma$ on $\mathcal{H}$,
$\Gamma=\Gamma^{\dag},\,\Gamma^{2}=1,$ such that 
$\Gamma D+D\Gamma=0$ and $\Gamma a-a\Gamma=0$ for any  $a\in\mathcal{A}$.
It is \emph{real} if there is an
antilinear isometry $J$ (called real structure) which satisfies 
\begin{equation}
J^{2}=\epsilon\mathbb{I},\quad JD=\epsilon'DJ,\quad J\Gamma=\epsilon''\Gamma J\label{eq:4}
\end{equation}
where $\epsilon,\epsilon',\epsilon''\in\left\{ 1,-1\right\} $ define
the $KO$-dimension (see e.g. \cite{Connes:2008kx} for details).

The real structure implements an action of the opposite algebra{\footnote{Identical to $\mathcal{A}$ as a vector space, but with reversed product:
$a^{\circ}b^{\circ}=(ba)^{\circ}$.}} $\mathcal{A}^{\circ}$, obtained by
identifying $Jb^{*}J^{-1}$ with $b^{\circ}\in\A^\circ$ (for any
$b\in\A$), which is
asked to commute with $\mathcal{A}$:
\begin{equation}
[a,JbJ^{-1}]=0\quad\forall\;a,b\in\mathcal{A}.
\label{eq:12}
\end{equation}
This is called the \emph{order zero condition} and it permits one to
define a right action of $\A$ on~$\HH$
\begin{equation}
  \label{eq:14}
  \psi a := a^\circ \psi = Ja^*J^{-1}\psi \quad \forall \psi\in\HH.
\end{equation}
 Another condition that plays an important role is the
\emph{order one condition} \cite{Connes:1996fu}, which for twisted
spectral triples one writes as\cite{buckley,Lett.}
\begin{equation}
[[D,a]_{\rho},JbJ^{-1}]_{\rho_{0}}=0\quad\forall a,b\in\mathcal{A}\label{eq:13},
\end{equation}
where\footnote{
$\rho_{0}$ is the
``natural'' image of $\rho$ in the automorphism group of
$\mathcal{A}^{\circ}$: $\rho_\circ(b^\circ) = (\rho(b))^\circ$.} 
\begin{equation}
\rho_{0}(JbJ^{-1}):=J\rho(b)J^{-1}.\label{eq:17}
\end{equation}
The usual conditions for untwisted spectral triples are retrieved by taking
$\rho$ to be the identity  automorphism $\rho(a)=a$. 

A gauge theory is described - in its non twisted
 version -  by an \emph{almost commutative geometry},
that is the product
\begin{equation}
\mathcal{A}=\cinf\otimes\A_{F},\;\HH=L^{2}(\M,S)\otimes\HH_{F},\;D=\slashed{\partial}\otimes{\mathbb{I}}_{F}+\gamma_E\otimes D_{F}\label{eq:02}
\end{equation}
of the canonical spectral triple
$\left(C^{\infty}(\M),L^{2}(\M,S),\slashed{\partial}\right)$
associated to an (oriented closed) Riemannian spin manifold $\M$ of even dimension\footnote{In this paper we will consider only manifolds of even dimension. The odd case has technical issues which we prefer to ignore. For a full discussion of product geometries see~\cite{Farnsworth:2016qbp} and the references therein.}
$m$, with a finite dimensional spectral triple $
(\A_{F},\HH_{F},D_{F})$.
 Recall that $L^{2}(\M,S)$ denotes the Hilbert space of square integrable spinors on $\M$,
on which $\cinf$  acts as 
\begin{equation}
  \label{eq:19}
  \pi(f)=f\left(
    \begin{array}{cc}
      \I_{\frac n2}& 0 \\ 0&   \I_{\frac n2}
    \end{array}\right)\quad\quad f\in\cinf,
\end{equation} 
with  $n=2^{\frac m2}$
the dimension of the spin representation. The Dirac operator is 
\begin{equation}
\slashed{\partial}=-i\sum_{\mu=1}^{m}\gamma^{\mu}_E\nabla_{\mu}^{S}
\quad\text{ where }\quad \nabla_{\mu}^{S}=\del_{\mu}+\omega_{\mu}^{S}
\label{eq:80}
\end{equation}
with $\gamma_E^{\mu}={\gamma_E^{\mu}}^{\dagger}$
the selfadjoint Euclidean-signature Dirac matrices and $\omega_{\mu}^{S}$ the spin connection.
 The grading
 \begin{equation}
\gamma_E = \text{diag} (
       \I_{\frac n2},  - \I_{\frac n2})
\label{eq:86}
            \end{equation}
is the product of the Dirac matrices (it is usually called $\gamma^5$ in the physics literature).

To describe the standard model, under natural assumptions on the representation of the algebra, it is shown in \cite{Chamseddine:2008uq}
that the finite dimensional algebra in \eqref{eq:02} has to be
\begin{equation}
\A_F:=\bb C\oplus\bb H\oplus M_{3}(\bb C),\label{eq:115}
\end{equation}
acting on the finite
dimensional Hilbert space.
\begin{equation}
\mathcal{H}_{F}=\mathcal{H}_{R}\oplus\mathcal{H}_{L}\oplus\mathcal{H}_{R}^{c}\oplus\mathbb{\mathcal{H}}_{L}^{c}=\bb C^{96}\label{eq:56}
\end{equation}
where ${\cal H}_R={\mathbb C}^8\times {\mathbb C}^3$ is spanned by the $N=3$ generations of
8 right-handed fermions
(electron, neutrino, up and down quarks with three colors each), ${\cal
H}_L$ stands for left fermions, and the exponent $^c$ is for the antiparticles.
The finite dimensional Dirac operator $D_{F}$ is a $96\times96$ matrix whose entries are the Yukawa couplings of fermions,
the Dirac and Majorana masses of neutrinos, the Cabibbo matrix and the mixing matrix
for neutrinos.

The twisted spectral triple  of the Standard Model is obtained by
making
\begin{equation}
  \cinf\otimes {\bb C}^2\simeq \cinf\oplus\cinf
\end{equation}
act on $L^2(\M,S)$ as 
\begin{equation}
  \label{eq:43}
  \pi((f,g)):=\left(
 \begin{array}{cc}  f\I_{\frac n2}&0\\ 0& g\I_{\frac n2}  
    \end{array}
\right) \quad\quad \forall (f,g)\in \cinf\oplus \cinf,
\end{equation}
choosing as automorphism
\begin{equation}
  \label{eq:44}
  \rho((f,g)) = (g,f) \quad\quad \forall (f,g)\in\cinf\oplus\cinf.
\end{equation}
The complete twisted spectral triple thus
consists in
\begin{equation}
  \label{eq:42}
  \left((\cinf\otimes \A_F)\otimes {\bb C}^2, L^2(\M,S)\otimes {\bb C}^{96},
  D=\ds\otimes \I_{96} + \gamma_E\otimes D_F;\, \rho\right),
\end{equation}
where the ``doubled'' algebra
 $(\cinf\otimes \A_F)\otimes {\bb C}^2$
acts on $\HH= L^2(\M,S)\otimes {\bb C}^{96}$ as in the Standard Model,
except that the representation \eqref{eq:19} of $\cinf$ is substituted
with the representation \eqref{eq:43} of $\cinf\otimes{\mathbb C}^2$. 
The Dirac operator
is unchanged.

The grading $\Gamma$ and the real structure $J$ are as in the non
twisted case, namely  
\begin{equation}
\label{eq:046}
\Gamma=\gamma_E\otimes\gamma_{F}\quad\text{ where
}\quad\gamma_{F}:=\text{diag }(
\mathbb{I}_{8N}, -\mathbb{I}_{8N}, -\mathbb{I}_{8N},\mathbb{I}_{8N}),
\end{equation}
and
\begin{equation}
J=\mathcal{J}\otimes J_{F} \quad\text{ where }\quad  J_{F}:=\left(\begin{array}{cc}
0 & \mathbb{I}_{16N}\\
\mathbb{I}_{16N} & 0
\end{array}\right)cc
\label{eq:46}
\end{equation}
with ${\cal J}=i\gamma^{0}\gamma^{2}cc$
the charge conjugation on $L^{2}(\M,S)$
(with $cc$ the complex conjugation).

The fermionic fields are elements of the Hilbert space $\HH$. The bosonic fields
are obtained by the so-called \emph{twisted fluctuations of $D$ by
  $\A$}, which amount to substituting $D$  with \cite{buckley}
\begin{equation}
D_{A_{\rho}}:=D+A_{\rho}+\epsilon'\,JA_{\rho}J^{-1}\label{eq:3}
\end{equation}
where 
$A_\rho$ is an element of the set of twisted generalized one forms~\cite{Connes:1938fk}
\begin{equation}
\Omega_{D}^{1}(\mathcal{A},\rho):=\left\{ \sum_{i}a_{i}[D,b_{i}]_{\rho},\;a_{i},b_{i}\in\mathcal{A}\right\} .\label{eq:2}
\end{equation}
\begin{rem}
\label{rem:intro1}
The twisted fluctuations of the
Dirac operator, which were initially introduced by ana\-logy with the non twisted case~\cite{buckley}, have now in~\cite{Landi:2017aa} been placed on the same rigorous footing as
Connes' original ``fluctuations of the metric'' \cite{Connes:1996fu},
namely as a way to export a real twisted spec\-tral triple to a Morita
equivalent algebra. In particular, in  case of self-Morita equivalence,
one obtains formula \eqref{eq:3}.
Additionally, a gauge
transformation is implemented as in the non-twisted case, namely as a
change of connection in the bimodule that implements Morita
equivalence. This yields formula \eqref{eq:6}, which is our starting point in this paper.

Note however, that
there is an important difference between the twisted and the
non-twisted cases: while usual fluctuations preserve the
selfadjointness of the Dirac operator, twisted-fluctuations may not. In~\cite{Landi:2017aa} this issue was addressed working out the necessary and sufficient
conditions, such that the unitary $u$ which implements the twisting automorphism (in
case the latter lifts to an inner automorphism of ${\cal B}(\HH)$) must satisfy in order to preserve selfadjointness.  Interestingly, there are other solutions beyond the obvious ones 
(i.e.\  $u$ invariant under the twist). 

In this paper we provide an alternative solution: instead of trying
to preserve selfadjointness, we investigate whether there is a ``more natural''
property  preserved under a twisted
fluctuation. We find one: selfadjointness with
respect to the inner product induced by the twist. Unexpectedly, in
the
case of the twisted spectral triple of the Standard Model, the
induced product is the
Krein product of Lorentz spinors.
It is quite remarkable that the Lorentz
structure emerges from the algebraic properties of the Euclidean spectral triple
and its twist (all the more that the later is unique, under the
condition that the fermionic sector of the theory is untouched by the
twist\cite{Landi:2017aa}).
\end{rem}

When $D_{A_{\rho}}$ is selfadjoint, we call
it a \emph{twisted-covariant Dirac operator}. One then shows \cite{buckley} 
that 
$(\mathcal{A},\mathcal{H},D_{A_{\rho}};\rho)$ is a real twisted spectral
triple, with the same real structure and $KO$-dimension as $(\mathcal{A},\mathcal{H},D;\rho)$.

A gauge transformation for a twisted spectral
triple  \cite{Landi:2017aa}  is implemented by
the simultaneous action on $\HH$ and ${\cal L}(\HH)$ (the space of linear operators on
$\HH$) of
the group of unitaries of $\A$,
\begin{equation}
  \label{eq:10}
  {\cal U}(\A):=\left\{ u\in\A, u^*u = uu^*=\bb I \right\}.
\end{equation}
The action on $\HH$ follows from 
the adjoint action of $\A$ (on
the left via its representation, on the right by
\eqref{eq:14}), that is
\begin{equation}
  \label{eq:24}
  \text{Ad}(u) \psi = u\psi u^* = u JuJ^{-1}\psi \quad \forall
  \psi\in\HH,\, u\in {\cal U}(\A).
\end{equation}
The action on ${\cal L}(\HH)$ is defined as 
\begin{equation}
T\mapsto \text{Ad}(\rho(u)) \,T \, \text{Ad}(u^*) \qquad \forall T\in {\cal L}(\HH),
\label{eq:59}
\end{equation}
where
\begin{equation}
  \label{eq:34}
  \text{Ad}(\rho(u))= \rho(u) J\rho(u)J^{-1}.
\end{equation}
 In particular, for $T=D_{A_\rho}$ a twisted covariant Dirac operator
 \eqref{eq:3},
 one has (see details in appendice, and also \cite{Landi:2017aa})
 \begin{equation}
    \label{eq:6}
    \text{Ad}(\rho(u))\, D_{A_\rho} \text{Ad}(u^*) = D_{A_\rho^u}
  \end{equation}
  where
  \begin{equation}
    \label{eq:7}
    A_{\rho}^u := \rho(u)A_\rho u^{*}+\rho(u)\left[D,u^{*}\right]_{\rho}.
  \end{equation}
The map $A_\rho\mapsto A_\rho^u$
is a twisted version of the usual law of transformation of the gauge
potential in noncommutative geometry \cite{Connes:1996fu}. 
\bigskip

The interest of twisting the
spectral triple of the Standard Model is that whereas the part $D_R$ 
of the operator  $D_F$ that contains the Yukawa
coupling $k_R$ of the right handed neutrino is transparent to usual
inner fluctuations,  
\begin{equation}
  [D_{R},a]=0 \quad  \forall a\in\A_F,
\label{eq:54}
  \end{equation}
it is not transparent to
  twisted inner fluctuations \eqref{eq:3},
  \begin{equation}
[D_R,a]_\rho\neq 0 \quad \text{ for some }a\in\A_F\otimes{\mathbb C}^2.
\label{eq:70}
  \end{equation}
That $D_R$ did not fluctuate remained almost unnoticed until the observation in 
\cite{Chamseddine:2012fk} that turning the (constant) $k_R$ 
into a field $k_R\sigma$ provides precisely the extra scalar field required to
stabilize the electroweak vacuum, and also provides a way of naturally
accommodating the mass of the Higgs boson.
The non-fluctuation of $D_R$ by internal symmetries can be traced back to the first-order
condition (as already noticed in  \cite{Connes:2010fk}).
To justify the substitution $k_R\rightarrow k_R\sigma$, various
solutions have been proposed:
\begin{itemize}
\item[-] Make the first-order condition more flexible, as
  investigated in \cite{Chamseddine:2013fk, Chamseddine:2013uq}, with
  phenomenological consequences in \cite{Chamseddine:2015aa} (see also
\cite{U.-Aydemir:2015aa}).

\item[-] Attempt to fluctuate the $\sigma$ field using the outer
  symmetries of the theory, as initiated in~\cite{Farnsworth:2014vva},
  leading to a minimal and phenomenologically viable Standard Model
  extension.

\item[-]  Double the algebra and twist the first order condition, as
  explained above.
\end{itemize}

\section{Twist and Lorentz Structure}
\label{sec:twistlorentz}

We show that when the automorphism $\rho$ in a
twisted spectral triple $(\A, \HH, D; \rho)$ is inner,
then there exists a natural $\rho$-twisted inner product on
$\HH$. Furthermore, for the twisted geometry of the
Standard Model \eqref{eq:02}, this
inner-product is a Krein product of  Lorentzian spinors. 

\subsection{Twisted inner product}
\label{subsec:twistinner}
Let $\HH$ be an Hilbert space with inner product
  $\langle \cdot , \, \cdot \rangle$, and $\rho$ be an automorphism of   $\cal B(\HH)$.
\begin{df}
\label{de:twist}
A  $\rho$-twisted inner product
  $\langle\cdot ,\, \cdot \rangle_\rho$ is an inner product on $\HH$ such
  that
   \begin{equation}
    \langle\Psi,\mathcal{O}\Phi\rangle_{\rho}=\langle\rho(\mathcal{O})^\dag\Psi,\Phi\rangle_{\rho}\qquad
    \forall {\cal O}\in{\cal B}(\HH),\; \Psi,\,\Phi\in\HH,
    \label{inner product}
  \end{equation}
  where $\rho({\cal
    O})^\dag$ is the adjoint of $\rho({\cal
    O})$ with respect to the initial Hilbert inner product
  $\langle\cdot,\cdot  \rangle$.
\end{df}
\noindent  We denote 
\begin{equation}
{\cal O}^+:= \rho({\cal O})^\dag
\label{eq:26}
\end{equation}
the  adjoint of a bounded operator ${\cal O}$ with respect to the $\rho$-twisted
inner product. For short we call the later the $\rho$-product, and ${\cal
  O}^+$ the $\rho$-adjoint of ${\cal O}$.
An operator $\mathcal{O}$ (resp. $U$) on $\mathcal{H}$ is said
\textit{$\rho-$hermitian} (\textit{$\rho-$unitary})
if it is selfadjoint (unitary) with respect to the $\rho$-product:
${\cal O}^+={\cal O}$, $U^+U=UU^+=\I$. In terms of the initial
Hilbert product on $\HH$, this reads
\begin{equation}
\mathcal{O}=\rho(\mathcal{O})^\dag,\quad \rho(U)^\dag U = U\rho(U)^\dag=\I.\label{eq:hermitianity}
\end{equation}

If  $\rho$ is an inner automorphism of ${\cal B}(\HH)$,  such that there
exists a unitary operator $R$ on $\HH$ satisfying
\begin{equation}
  \label{eq:27}
  \rho({\cal O}) = R{\cal O} R^\dag \qquad \forall {\cal O}\in{\cal B}(\HH),
\end{equation}
then a natural $\rho$-product is 
\begin{equation}
\langle\Psi,\Phi\rangle_\rho=\langle\Psi, R\Phi\rangle=\langle R^\dag\Psi, \Phi\rangle.
\label{eq:28}
\end{equation}
Indeed, for any bounded operator $\cal O$ on $\HH$ one checks that 
\begin{align}
  \label{eq:29}
  \langle\Psi, {\cal O}\Phi\rangle_\rho&=\langle\Psi, R{\cal
                                         O}\Phi\rangle =  \langle
                                         {\cal O}^\dag R^\dag\Psi, \Phi\rangle
= \langle {\cal O}^\dag R^\dag\Psi, R^\dag R\Phi\rangle\\
&=\langle R{\cal O}^\dag R^\dag \Psi, R\Phi\rangle=\langle \rho({\cal
  O}^\dag )\Psi, \Phi\rangle_\rho=\langle\rho({\cal O})^\dag \Psi, \Phi\rangle_\rho=\langle {\cal O}^+\Psi, \Phi\rangle_\rho
\end{align}
where we used that an inner automorphism is necessarily a
$^*$-automorphism, that is
\begin{equation}
\rho(O)^\dag = (R{\cal O}R^\dag)^\dag = R{\cal O}^\dag R^\dag =\rho({\cal O}^\dag).
\label{eq:9}
\end{equation}
Notice that $R$ is both unitary (by definition) and $\rho$-unitary:
one has $\rho(R) = RRR^\dag=R$ so that 
$\rho(R)^\dag R = R^\dag R =\I=R\,\rho(R)^\dag$.

All the definitions above
extend to any (i.e. non
necessarily bounded) densely defined linear
operator $T$ on $\HH$ whose domain is globally invariant under left
multiplication by $R^\dagger$: if
$R^\dagger\psi \in \text{Dom}(T)$ for any  $\psi\in \text{Dom}(T)$,
then one defines the
action of $\rho$ on $T$ as
\begin{equation}
\rho(T) \Psi := R \,T \,R^\dagger\Psi.
\label{eq:15}
\end{equation}

If in addition $T$ is adjointable (i.e. the operator $T^\dagger$ defined by $\langle T^\dagger \eta,
\xi\rangle := \langle \eta, T\xi\rangle$ for all $\xi\in
\text{Dom}(T)$, is densely defined), then so is $RTR^\dagger$ and by
\eqref{eq:26} one defines the $\rho$-adjoint of $T$ as
$T^+:= \rho(T)^\dagger$.
The operator $T$ is $\rho$-hermitian when $T^+=T$.

The extension of an inner automorphism $a\to uau^*$ of $\A$ to an automorphism of ${\cal B}(\HH)$ is not unique (just consider
two distinct unitaries $R_1$, $R_2$ in ${\cal B}(\HH)$ such that $R_1
a R_1^\dag = R_2 a
R_2^\dag$ for any $a\in \A$). Any such extension defines 
an automorphism of $\A_\circ$: 
\begin{equation}
  \label{eq:38}
  \rho(Ja^*J^{-1}) = R\, J a^* J^{-1} R^\dag\quad \forall a\in\A.
\end{equation}
We say that an inner automorphism is compatible with
  $J$ if it admits an extension such that \eqref{eq:38}
  agrees with $\rho_\circ\in\text{Aut}(\A_\circ)$ defined in
  \eqref{eq:17}. More precisely:
\begin{df}   
\label{def:compreal}
 Given a real spectral triple $(\A, \HH, D)$,  an inner automorphism $\rho$ of $\A$ is compatible with
  the real structure $J$ if there exists a unitary $R\in\cal B(\HH)$
  such that
  \begin{align}
\rho(a) = RaR^\dag \quad \text{ and }\quad
    \label{eq:39bis}
     J\,R\, a^*\,R^\dag J^{-1} = R J a^* J^{-1} R^\dag \qquad \forall a\in\A.
  \end{align}
\end{df}

\noindent This condition  is verified in particular when the inner
automorphism can be implemented by a unitary $R$ such that

\begin{equation}
  \label{eq:40}
  J R = \pm  RJ .
\end{equation}

\subsection{Lorentzian signature and Krein space} 
\label{subsec:lorentzkrein}

    In definition \ref{de:twist} we do not require the $\rho$-product
    to be positive definite. Since $R$ is unitary, one has that
    $\langle\cdot ,\, \cdot \rangle_\rho$ is non-degenerate. If in
    addition we impose that $R$ is selfadjoint and different from the
    identity (i.e. the automorphism $\rho$ is not the trivial one),
    then $R$ has eigenvalues $\pm 1$. The two corresponding eigenspaces
    $\HH_+, \HH_-$ are such that $\HH= \HH_+ \oplus \HH_-$ and the
    $\rho$-product is positive definite on $\HH_+$, negative definite
    on $\HH_-$. In other words, a space $\HH$ equipped with
    the product $\langle \cdot, \cdot\rangle_\rho$ is a Krein space. Furthermore, the
    operator $R$ is a fundamental symmetry, that is it satisfies
    $R^2=\I$ and the inner product $\langle\cdot , R\cdot\rangle_\rho$
    is positive definite  on $\HH$ (in our case, this is simply the
    Hilbert product one started with).

In the twisted spectral triple of the Standard Model, the
automorphism $\rho$ is the flip \eqref{eq:44}. It is implemented on $L^2(\M,S)$ by
the adjoint action of the selfadjoint unitary operator
\begin{equation}
R=\begin{pmatrix} 0 &\mathbb 1_2\\
\mathbb 1_2 & \mathbb 0 \end{pmatrix}. 
\label{eq:53}
\end{equation}
 This matrix has eingenvalues $\pm 1$, hence $\HH$ equipped with the
 $\rho$-product is a Krein space. 
The Euclidean Dirac matrices in the chiral basis are,
for $\mu=0,j$ with $j=1,2,3$,
\begin{equation}
  \label{eq:25}
  \gamma^\mu_E=\left(
    \begin{array}{cc}
      0& \sigma^\mu\\ \tilde\sigma^\mu &0 
    \end{array} \right)\quad \text{ where } \sigma^\mu = \left\{\I_2,
      -i\sigma^i\right\},\; \tilde\sigma^\mu=\left\{\I_2, i\sigma^i\right\}
\end{equation}
where $\sigma_j$, $j=1,2,3$ are the Pauli matrices. Thus
$R$ is nothing but
$\gamma^0_E$ and the $\rho$-product~\eqref{inner product} is now the usual inner product
of quantum field theory in Lorentz signature, where instead of
$\psi^\dagger$ it appears $\bar\psi=\psi^\dagger \gamma_E^0$:
\be \langle\psi,\phi\rangle_\rho=\langle \psi, \gamma_E^0\phi\rangle
= \int d^{4}x\psi^{\dagger}\gamma_E^{0}\phi:=\int
d^{4}x\bar\psi\phi.
\label{eq:Minkosky} \ee
 
Furthermore, one has
\begin{equation}
  \label{eq:35}
  \rho(\gamma_E^0)=(\gamma_E^0)^3 = \gamma^0_E,\quad
  \rho(\gamma_E^j)=\gamma_E^0 \gamma_E^j \gamma_E^0= -\gamma^j_E.
\end{equation}
The twist therefore performs some sort of Wick rotation
  whereby the sign of the time-component Dirac matrix is changed with respect to
  the spatial directions. The matrix $R$ is of course expressed in  a
  particular basis, and its action on the Clifford algebra generated
  by the $\gamma$'s  singles out one particular direction, which we
  identify with time.  Then it makes sense to define the integral on a
  time slice and have fields normalized only for the space integral,
  which is what is commonly done. 
However, the $\rho(\gamma^i_E)$'s are
  not the Lorentzian signature (i.e. Minkowskian) gamma matrices,
  \begin{equation}
    \label{eq:63}
    \gamma^{0}_M =  \gamma^{0}_E\;,\quad \gamma_{M}^{j}= i\gamma_{E}^{j}
    \quad j=1,2,3.
  \end{equation}
Viewing the Wick rotation as the operator:
$W:\gamma^\mu_E\to\gamma^\mu_M$, that is
\begin{equation}
 W(\gamma^0_E)=\gamma^0_E,\quad
 W(\gamma^j_E)=i\gamma^j_E,\label{eq:16}
 \end{equation}
one has that the twist \eqref{eq:35} is the
 square of the Wick rotation
 \begin{equation}
   \label{eq:64}
   \rho(\gamma^0_E)= W(W(\gamma^0_E)),\quad  \rho(\gamma^j_E)= W(W(\gamma^j_E)).
 \end{equation}

The Euclidean Dirac matrices are selfadjoint for the Hilbert product
of $L^2(\M,S)$, but (except for $\gamma^0_E$) not $\rho$-hermitian since from
\eqref{eq:26} and \eqref{eq:35} one has
\begin{equation}
  \label{eq:65}
  (\gamma^j_E)^+ = \rho(\gamma^j_E)^\dag= -{\gamma^j_E}^\dag.
\end{equation}
On the contrary, the Minkowskian gamma matrices (except $\gamma^0_M$)  are not selfadjoint for
the Hilbert product since \eqref{eq:63} yields 
$(\gamma^j_M)^\dag = - \gamma^j_M$; but they are $\rho$-hermitian since
\begin{equation}
  \label{eq:67}
  \rho(\gamma^j_M) = i\rho(\gamma^j_E) = -i\gamma^j_E= -\gamma^j_M,
\end{equation}
so that
\begin{align}
(\gamma_M^j)^+= \rho(\gamma^j_M)^\dag  &=(\gamma_0 \gamma_M^j \gamma_0)^\dag
 =\gamma_0 (\gamma_M^j)^\dag\gamma_0 = -\gamma_0 \gamma_M^j\gamma_0 =-\rho(\gamma_M^j)=\gamma^j_M.
\label{eq:68}
  \end{align}
The ``temporal'' gamma matrix $\gamma^0:=\gamma^0_E=\gamma^0_M$ is both
selfadjoint and $\rho$-hermitian.
 
The twist naturally defines a Krein structure, while maintaining in
     the background the Euclidean structure.  Applications of
     Krein spaces to noncommutative
 geometry framework  have been recently
    studied in\cite{VanDenDungen,BiziBrouderBesnard} as well as in
    \cite{Franco:2012fk,Franco:2014fk,Franco:2012fkbis} (see reference
    therein for earlier attempts of adapting Connes noncommutative
    geometry to the Minkowskian signature).

\begin{rem}
\label{rem:intro2}
Wick rotations are a delicate issue in the presence of fermions, and have been discussed in detail in~\cite{DAndrea:2016aa} (see also~\cite{Kurkov:2017wmx}). There are intimate relations with the projection necessary to remove the extra degrees of freedom given by fermion doubling~\cite{Lizzi:1996vr}. In the present paper we concentrate on the issue of the change of signature, and how it relates with the twist. The transition from an Euclidean signature theory which contains fermions involves changes of degrees of freedom, usually a doubling, as occurred for example in~\cite{Osterwalder:1973zr,vanNieuwenhuizen:1996tv}. In the present approach there is actually a quadruplication of degrees of freedom, obtained as a product of two duplications. One of these  disappears after Wick rotation, while the other must be projected out. The nature of the duplications seen in this work differ slightly from those cited above.
\end{rem}

\section{Actions}
\label{sec:action}

For an almost commutative geometry \eqref{eq:02} with grading $\Gamma$, the fermionic action is \cite{Chamseddine:2007oz}
\begin{equation}
  \label{eq:11}
 S^F(D):=  \langle J\tilde\psi, D\tilde \psi\rangle
\end{equation}
where $\psi$ is a vector in the even part of the Hilbert space 
\begin{equation}
  \label{eq:55}
  \HH_e:=\left\{\psi\in\HH, \gamma\psi = \psi\right\}
\end{equation}
which, seen as an operator on the Fock space, is a
Grassmanian variable $\tilde\psi$.
The bosonic action is \cite{Chamseddine:1996zu}
\begin{equation}
S^B(D):=\text{Tr}\,f\left(\frac{D^2}{\Lambda^2}\right)
\label{eq:41}
\end{equation}
where $\Lambda$ is an energy cutoff and $f$ a smooth approximation of
the characteristic function of the interval~$[0,1]$. Both actions are invariant
under a gauge transformation \cite{Chamseddine:2007oz}, that is the simultaneous transformation 
\begin{equation}
D\mapsto \text{Ad}(u) \,D\, \text{Ad}(u^*) \quad \text{ and } \quad  
\psi\longmapsto  \text{Ad}(u)\psi.
\label{eq:61}
\end{equation}
 For the almost
 commutative geometry \eqref{eq:02} with $\M$ a Riemannian spin
 manifold and $(\A_F, \HH_F, D_F)$ as described in \S \ref{sec:SM}, $S^F$
 yields the fermionic action of the Standard Model, while the asymptotic expansion
 of $S^B$ yields the bosonic part, including the Higgs, together with the
 Einstein-Hilbert action (in Euclidean signature) and an extra Weyl
 term.

For a twisted spectral triple, both actions \eqref{eq:11} and 
 \eqref{eq:41} are well defined, but their invariance under a gauge
 transformation, that is - as explained at the end of \S \ref{sec:SM} - the 
 simultaneous transformation
\begin{equation}
  \label{eq:62}
  D\mapsto D_{A_\rho^u}:=\text{Ad}(\rho(u)) \,D\, \text{Ad}(u^*) \quad \text{ and } \quad  
\psi\longmapsto  \text{Ad}(u)\psi,
\end{equation}
is not guaranteed:
\begin{itemize}
\item  $S^F(D)$ has no reason to be invariant, unless $u$
  is invariant under the twist: $u=\rho(u)$;

  \item  $S_B(D_{A_\rho^u} )$ is defined only if 
    $D_{A_\rho^u}$ is selfadjoint, or at least normal (which guarantees
    that $f\left(\frac{D_{A_\rho^u}^2}\Lambda\right)$ makes sense by the spectral theorem). But as has already been discussed in remark
    \ref{rem:intro1}, for an arbitrary unitary $u$, the operator $D_{A_\rho^u}$ has no reason to be
    selfadjoint, nor even normal. Notice however that it has compact
    resolvent, being
    a perturbation by bounded operator of the selfadjoint compact-resolvent
    operator $D$. Thus as soon as it is normal, the trace in
    \eqref{eq:41} is finite for any value of the cutoff
    $\Lambda$ and the bosonic action $S^B(D_{A_\rho^u})$ is then well defined. 
\end{itemize}

So as to ensure that both  the fermionic and bosonic actions remain well defined and invariant under  gauge transformations, one may restrict the range of acceptable unitaries to those such that  \eqref{eq:11} is invariant under \eqref{eq:61} and such that $D_{A_\rho^u}$ is normal. It is
likely however that there is no other solution other than the unitaries invariant
under the twist,  which would drastically limit the interest of the whole construction: why introducing a twist only to ignore it at the end? Another
solution, and the one that we explore in the following, is to 
take advantage of the twisted inner
product \eqref{eq:28} so as to modify the definitions of $S^F$ and
$S^B$, in order to obtain actions invariant under \eqref{eq:62} for \emph{any} unitary~$u$.

\subsection{Fermionic Action}
\label{subsec:fermact}

\noindent To build a fermionic action for a twisted spectral triple that is
invariant under the gauge transformation \eqref{eq:62}, it suffices to substitute the
Hilbert inner product in \eqref{eq:11} with the $\rho$-product.
%
\begin{prop}
\label{Prop:fermion1} Let $(\A, \HH, D;\rho)$ be a real twisted spectral triple with
  $\rho$ an inner automorphism of ${\cal B}(\HH)$ compatible with the
  real structure in the sense of Def.~\ref{def:compreal}. Then 
\begin{equation}
  \label{eq:47}
  \mathfrak A_D^\rho(\psi, \phi) = \langle J\psi, D\phi\rangle_\rho\quad
  \forall \psi, \phi\in\text{Dom} \,D
\end{equation}
 is a bilinear form invariant under the simultaneous transformations
of  \eqref{eq:62}. That is
 \begin{equation}
   \label{eq:49}
    \mathfrak A_D^\rho(\psi, \phi) =     \mathfrak A_{\text{Ad}(\rho(u))\,D\, \text{Ad}(u^*)}^\rho\left(\text{Ad}(u)\psi,\, \text{Ad}(u)\phi\right)\quad
  \forall \psi, \phi\in\text{Dom} \,D,\; u\in{\cal U}(\A).
 \end{equation}
\end{prop}
\begin{proof}
One simply adapts to the twisted case the proof of  \cite[Prop. 1.213]{Connes:2008kx}.
The Hilbert product on $\HH$ is antilinear on the first variable, and
the same is true for the $\rho$-product. Since $J$ is antilinear,  one
has that $\mathfrak A_D^\rho(\cdot, \cdot)$ is linear in both
variables.

  Let $U=\text{Ad}(u)=uJuJ^{-1}$. By \eqref{eq:39bis} one has  
  \begin{equation}
    \label{eq:50}
 \text{Ad}(\rho(u)) = \rho(u) J\rho(u) J^{-1} = \rho(u) \rho(JuJ^{-1})
 = \rho(uJuJ^{-1})= \rho(U). 
  \end{equation}
Therefore
   \begin{align}
     \label{eq:48}
    \mathfrak A_{\rho(U)DU^*}^\rho(U\psi, U\phi) &=\langle JU\psi, \rho(U)D U^* U\phi\rangle_\rho =    \langle UJ\psi,
     \rho(U)D \phi\rangle_\rho \\ &=  \langle J\psi,
     U^+\rho(U)D \phi\rangle_\rho=  \langle J\psi,
     D \phi\rangle_\rho
   \end{align}
where in the first line we have used the fact that $J$ commutes with $U$,
\begin{equation}
  \label{eq:51}
  JU = J(uJuJ^{-1}) =JJuJ^{-1}u = \epsilon''  u(JJ^{-1})J^{-1} u =\epsilon''  uJ(J^{-1})^2 u=
  uJu = UJ,
\end{equation}
and the last line comes from \eqref{eq:26}:
\begin{equation}
  \label{eq:52}
  U^+\rho(U) = \rho(U)^*\rho(U) = \I.
\end{equation}
This proves the result.
\end{proof}
 
In the (usual) description of the Standard Model with a (non twisted) spectral triple,
it is
important for the bilinear form $\langle J\phi,
D\psi\rangle$ to be antisymmetric,
$\langle J\phi,
D\psi\rangle= - \langle J\psi,
D\phi\rangle$,
and not to vanish when restricted to the even part \eqref{eq:55} of the Hilbert
space. This  makes $S^F$
in \eqref{eq:11} vanishing if computed with usual spinors,
but gives the expected fermionic action when computed with Grassman
fermionic fields. One restricts to $\HH_e$ in order to solve the fermion
doubling problem  (see \cite[I.\S 16.2]{Connes:2008kx} for details).
In the twisted case, the bilinear form \eqref{eq:47} is not
necessarily antisymmetric. It is however, if one restricts to
\begin{equation}
  \label{eq:57}
    \HH_R:=\left\{\psi\in \text{Dom} \, D, R\psi = \psi\right\}.
\end{equation}

\begin{prop} 
\label{Prop:fermion2} Let $(\A, \HH, D; \rho)$ be a real twisted spectral triple
  for which  $\rho$ is compatible with the real structure  in the sense of
  \eqref{eq:40}.
Then   the bilinear form ${\frak A}^\rho_D$ is such that 
  \begin{equation}
    \label{eq:69}
    {\frak A}^\rho_D(\psi, \phi) = \epsilon \epsilon' \;    {\frak
      A}^\rho_D(\phi, \psi) \quad \forall \psi, \phi\in\HH_R. 
  \end{equation}
\end{prop}
\begin{proof}

By definition of antiunitary operator, $\langle J\phi,
J\psi\rangle = \langle \phi, \psi\rangle$ for any $\psi,
\phi\in\HH$. Thus
\begin{equation}
  \label{eq:81}
    {\frak A}^\rho_D(\psi, \phi) = \langle J\psi, RD\phi\rangle=
   \epsilon \langle J\psi, J^2RD\phi\rangle=  \epsilon \langle
   JRD\phi , \psi\rangle.
\end{equation}
Let $ JR = \epsilon''' RJ$, where $\epsilon'''=\pm 1$. For $\psi$, $\phi$ in $\HH_R$, one obtains
\begin{align}
  \label{eq:82}
   {\frak A}^\rho_D(\psi, \phi) &=\epsilon \epsilon''' \langle
   RJD\phi , \psi\rangle= \epsilon\epsilon' \epsilon''' \langle
   RDJ\phi , \psi\rangle=\epsilon\epsilon' \epsilon''' \langle
   J\phi , DR^\dag\psi\rangle,\\
  &=\epsilon\epsilon' \epsilon''' \langle   JR^\dag R\phi , D\psi\rangle = \epsilon\epsilon' \langle
   R^\dag JR\phi , D\psi\rangle= \epsilon\epsilon' \langle J\phi , RD\psi\rangle=  \epsilon\epsilon'  {\frak A}^\rho_D(\phi, \psi).
\end{align}
where in the second line we use $R^\dag R=\I$ and $R^\dag \psi=\psi$,
then $R^\dag J = \epsilon''' J R^\dag$.  
\end{proof}
\noindent
The twisted spectral triple of the Standard Model presented in \S
\ref{sec:SM} has $KO$-dimension $2$, with a twist $R=\gamma^0_E$ 
compatible with the real structure. Hence the above proposition
applies with $\epsilon=-1$, $\epsilon'=1$, meaning that 
${\frak A}^\rho_D$ is antisymmetric as expected. 

Formally, the twisted fermionic action \eqref{eq:47} is similar to the non twisted
one \eqref{eq:11}: it is gauge invariant and symmetric on a given subspace of $\HH$.
However this subspace is $\HH_R$, not $\HH_e$. This may have some consequences on the physical
contents of the action, the study of which we leave to future work, as it will require first to compute the twisted fluctuations of the full Dirac
operator $D$ in \eqref{eq:02} (in \cite{buckley} was considered only the part
$D_R$ of $D_F$ that contain the Yukawa coupling of the right handed
neutrino).
\medskip

Alternatively, the Krein structure induced by
the twist suggests a way to define a fermionic
action which is antisymmetric on the whole of $\HH$: instead of restricting to 
$\psi\in\HH_R$, the alternative is to assume that $D$ is $\rho$-hermitian. This goes beyond the definition of  twisted
  spectral triples (which deals only with selfadjoint
  operator $D$). However, there exists several
  proposals for generalizing  spectral triples to the Lorentzian
  signature (see \cite{VanDenDungen} for the most recent one, and
  references therein). We do not intend to develop here a theory of \emph{Lorentzian twisted spectral
    triples}, we shall simply consider $(\A, \HH,D; \rho)$ that has all
  the properties of a real twisted spectral triple, with real
  structure $J$, except that $D$ is not
  selfadjoint but $\rho$-hermitian.
\begin{prop}
\label{prop:fermion3}
  Let $(\A, \HH, D; \rho)$ be as explained above, with $\rho$
  implemented by a selfadjoint unitary $R$ such that $JR=\epsilon''' RJ$.
Then the bilinear form ${\frak A}^\rho_D$ is such that
  \begin{equation}
    \label{eq:69bis}
    {\frak A}^\rho_D(\psi, \phi) = \epsilon \epsilon' \epsilon'''\;    {\frak
      A}^\rho_D(\phi, \psi) \quad \forall \psi, \phi\in
    \text{Dom} \, D. 
  \end{equation}
\end{prop}
\begin{proof}
By a similar calculation to that of  Prop. \ref{Prop:fermion2}, one arrives at
\begin{equation}
  \label{eq:21}
  {\frak A}_D^\rho(\psi, \phi) = \epsilon\epsilon'\epsilon'''\langle
  J\phi, D^\dag R^\dag\psi\rangle.
\end{equation}
The $\rho$-hermicity of $D$ together with $R=R^\dag$ implies $D^\dag R
= R D$. Hence
\begin{equation*}
   {\frak A}_D^\rho(\psi, \phi) = \epsilon\epsilon'\epsilon'''\langle
  J\phi, RD\psi\rangle=\epsilon\epsilon'\epsilon'''  \; {\frak A}_D^\rho(\phi, \psi).
\end{equation*}
 
\vspace{-1truecm}\end{proof}

As an illustration, consider the twisted spectral triple of the
Standard Model \eqref{eq:42} with $\ds$, $\cal J$ and $\gamma_E$
substituted with their  Lorentzian version,
\begin{equation}
\ds_M  := -i\gamma^\mu_M\partial_\mu,\quad  {\cal  J}_M := -i\gamma^2_M \, cc,\quad   \gamma_M:=\gamma^0_M\gamma^1_M\gamma^2_M\gamma^3_M= i^3
  \gamma^0_M\gamma^1_E\gamma^2_E\gamma^3_E = -i\gamma_E,
\label{eq:23}
\end{equation}
where $\gamma_M^\mu$
are the Minkovskian Dirac matrices \eqref{eq:63}.
One checks that 
\begin{equation}
    \label{eq:78}
D_M:=   \ds_M\otimes \I_{96} + \gamma_M\otimes D_F
  \end{equation}
 is $\rho$-hermitian for the 
  Krein structure induced by $\gamma^0_M$ since
  \begin{align*}
    (D_M)^+ = \rho((D_M)^\dag) &= \gamma^0 (\ds_M)^\dag\gamma^0 \otimes
    \I_{96} + \gamma^0_M\gamma_M^\dag\gamma^0_M\otimes D_F,\\
&=-i\gamma^0_M(\gamma^\mu_M)^\dag \gamma^0_M \partial_\mu +
  \gamma_M\otimes D_F= -i\gamma^\mu_M\partial_\mu +
  \gamma_M\otimes D_F=D_M,
  \end{align*}
where we use
$\ds_M^\dag=-i(\gamma^\mu_M)^\dag\partial_\mu$ then \eqref{eq:68},  and
   $ \gamma^0_M\gamma_M^\dag\gamma^0_M = i  \gamma^0_M\gamma_E\gamma^0_M = -i
    \gamma_E= \gamma_M$ coming from \eqref{eq:23} and the explicit
    forms \eqref{eq:86} and \eqref{eq:25} of $\gamma_E$ and
    $\gamma^0_M=\gamma^0_E$.
Moreover, denoting with an overbar the complex conjugate, one
  obtains from \eqref{eq:63} and \eqref{eq:23}
  \begin{equation}
    \label{eq:30}
    \overline{\gamma^0_M} =\gamma^0_M,\quad   \overline{\gamma^1_M} =\gamma^1_M,\quad   \overline{\gamma^2_M} =-\gamma^2_M,\quad   \overline{\gamma^3_M} =\gamma^3_M,
  \end{equation}
so that on a  Lorentzian four dimensional manifold, the real structure satisfies
\begin{align*}
({\cal J}_M)^2& = (-i
 \gamma_M^2\,cc)^2=\gamma_M^2\overline{\gamma_M^2}=
 -(\gamma^2_M)^2=\I;\\
  {\cal  J}_M \ds_M &= \ds_M{\cal J_M}\quad \text{ for }\quad{\cal  J}_M \ds_M - \ds_M{\cal J_M}= -\left(\gamma^\mu_M\gamma^2_M +
    \gamma^2_M\overline{\gamma^\mu_M}\right)\partial_\mu \, cc=0;\\
{\cal J}_M \gamma_M & =-i\gamma_M^2(\overline{\gamma^0_M}\,\overline{\gamma^1_M}\,\overline{\gamma^2_M}\,\overline{\gamma^3_M})cc
 = i\gamma_M^2(\gamma^0_M\gamma^1_M\gamma^2_M\gamma^3_M)cc  =-i
  (\gamma^0_M\gamma^1_M\gamma^2_M\gamma^3_M) \gamma_M^2 cc\\[4pt] &= \gamma_M{\cal J}_M.
\end{align*}
Since $J_F^2=\I$ (see \eqref{eq:46}), the first equation yields $({\cal
  J}_M\otimes J_F)^2=\I$, that is $\epsilon=1$. The second 
and third equations, together with $D_FJ_F=J_F D_F$ (coming from the
$KO$ dimension $6$  of the spectral triple \eqref{eq:02}) yield
\begin{align}
\nonumber
  D_M ({\cal J}_M\otimes J_F) &= \ds_M{\cal J}_M \otimes J_F +
  \gamma_M{\cal J}_M\otimes D_F J_F = {\cal J}_M \ds_M\otimes J_F +
  {\cal J}_M\gamma_M\otimes J_FD_F \\ &=  ({\cal J}_M\otimes J_F) D_M,
\end{align}
 so that $\epsilon'=1$. Finally one has $\gamma^0 {\cal J}_M =
 -{\cal J}_M\gamma^0$, meaning  $\epsilon'''=-1$. Therefore
 $\epsilon\epsilon'\epsilon'''=-1$, hence   ${\frak
  A}_D^\rho$ in Prop. \ref{prop:fermion3} is antisymmetric as expected.

The gauge invariance proved in Prop. \ref{Prop:fermion1} does not
depend on the selfadjointness of $D$, and thus is still valid for
$\rho$-hermitean $D$. What must be checked, however, for
Prop. \ref{prop:fermion3} to make sense is that a twisted perturbation
of a $\rho$-hermitean operator is still $\rho$-hermitean, and that
this property is preserved under gauge transformation.
\begin{prop}
\label{Prop:fermion4}
  Let $(\A, \HH, D;\rho)$ be as in proposition
  \ref{prop:fermion3}. Assume $\rho$ is compatible with the real
structure in the sense of \eqref{eq:40}. Then a twisted fluctuation
\begin{equation}
D_{A_\rho} =D +A_\rho + JA_\rho J^{-1}
\label{eq:31}
\end{equation}
of $D$ is $\rho$-hermitian as
long as $A_\rho$ is $\rho$-hermitian, i.e. 
\begin{equation}
  \label{eq:73}
  A_\rho=A^+ =\rho(A_\rho^\dag).
\end{equation} 
In addition, any gauge transform
\begin{equation}
  \label{eq:76}
  D'_{A_\rho} := \rho(U) D_{A_\rho} \, U^\dag\quad\text{ with }\;
  U:=\text{Ad}(u)\;\text{ for }\;
  u\in{\cal U}(\A)
\end{equation}
of a $\rho$-hermitian
operator $D_{A_\rho}$ is still $\rho$-hermitian.
\end{prop}
\begin{proof}
One has
\begin{align}
  \label{eq:71}
 \left(JA_\rho J^{-1}\right)^+ = R\left( JA_\rho  J^{-1}\right)^\dag R^\dag
  &=R\, JA_\rho^\dag   J^{-1}\, R^\dag\\ &= J\,RA_\rho^\dag R^\dag\,J^{-1} =J A_\rho^+ J^{-1} = JA_\rho J^{-1}.
\end{align}
Hence
\begin{equation}
  \label{eq:74}
  (D_{A_\rho})^+ = D^+ + A_\rho^+ +\left( J A_\rho J^{-1}\right)^+= D +
  A_\rho + JA_\rho J^{-1}=D_{A_\rho}.
\end{equation}

For the second claim, one has
\begin{align}
  \label{eq:75}
  \left(D'_{A_\rho}\right)^+ =   \rho\left({D'_{A_\rho}}^\dag\right) =
  \rho\left(UD_{A_\rho}^\dag\rho(U)^\dag\right) &= \rho(U)
  \rho\left(D_{A_\rho}^\dag\right) U^\dag\\
&= \rho(U) D_{A_\rho}^+ U^\dag= \rho(U) D_{A_\rho}  U^\dag= D'_{A_\rho}
\end{align}
where we use
\begin{equation*}
\rho(\rho(U)^\dag) = \rho\left((RUR)^\dag\right) =
\rho(R U^\dag R) =R^2 U^\dag R^2 = U^\dag.
\vspace{-1truecm}
\end{equation*} 
\end{proof}
\noindent Condition \eqref{eq:73} is the twisted version of the usual
requirement that the gauge potential $A$ in a fluctuation of the
metric should be selfadjoint.
\medskip

To summarize,  there are two candidates for the fermionic action:
  \begin{itemize}
  \item A Lorentzian one: $\langle J\psi,
    D^M_{A_\rho}\psi\rangle_\rho$, where $D^M_{A_\rho}$ is a
    $\rho$-hermitian twisted fluctuation of the Minkowskian operator ~\eqref{eq:78}.
\item A Euclidean one: $\langle J\psi, D^E_{A_\rho}\psi\rangle_\rho$
    with $\psi\in\HH_R$ and $D^E_{A_\rho}$ is a selfadjoint twisted
    fluctuation of the Euclidean operator \eqref{eq:02}.
     \end{itemize}
The Lorentzian action has been considered in \cite{Barrett:2007vf,VanDenDungen}.  The Euclidean action is similar to the one of the Standard Model
    \cite{Chamseddine:2007oz}, except that $\psi$ is in $\HH_R$
    instead of $\HH_e$. 

\subsection{Bosonic action}
\label{subsec:bosonact}

The easiest way to make the bosonic action \eqref{eq:41}
  well defined and invariant under a twisted gauge
transformation is to rewrite it as
\begin{equation}
  \label{eq:66}
\text{Tr} \,f\left(\frac{D^\dag D}{\Lambda^2}\right).
\end{equation}
Indeed, given a twisted spectral triple $(\A, \HH, D; \rho)$ (that is
  $D$ selfadjoint with compact resolvent), then under the map
\begin{equation}
D\rightarrow
\rho(U)DU^\dag,
\label{eq:72}
\end{equation}
one gets that $\frac{D^\dag D}{\Lambda^2}$
is mapped to $\frac{U D^\dag D U^\dag}{\Lambda^2}$ which has the same
trace as $\frac{D^\dag D}{\Lambda^2}$.
This is this action that has been
computed in \cite{buckley} for a selfadjoint twisted fluctuation
$D_{A_\rho}$ of the Dirac operator of the Standard Model.

If one considers instead a $\rho$-hermitian Dirac operator (with
compact resolvent),
\begin{equation}
D=D^+ = \rho(D)^\dag,\label{eq:79}
\end{equation}
then one can write 
\eqref{eq:66} in a twisted form (that is, without reference to the
Hilbert adjoint) as
\begin{equation}
  \label{eq:66bis}
\text{Tr} \,f\left(\frac{\rho(D) D}{\Lambda^2}\right).
\end{equation}
\noindent Taking for $D$ the $\rho$-Hermitian Minkowskian Dirac
operator $-i\gamma^\mu_M\partial_\mu$ (which has locally
  compact resolvent, see \cite [Prop.\ 4.2]{DungenRennie} and reference therein), it returns the Euclidean
action: by cyclicity of the trace, one
can substitute $\rho(D) D$ with $\frac 12\left(\rho(D) D +D\rho(D)\right)$, which is nothing but the Euclidean Laplacian  (up to a sign):
\begin{align}
\frac{1}{2}\left(\rho(D)D+D\rho(D)\right) & =\frac{1}{2}\left(i\gamma^{\mu\dagger}_M\partial_{\mu}i\gamma^{\nu}_M\partial_{\nu}+i\gamma^{\mu}_M\partial_{\mu}i\gamma^{\nu\dagger}_M\partial_{\nu}\right)\\
 & =-\frac{1}{2}\left(\gamma^{\mu\dagger}_M\gamma^{\nu}_M\partial_{\mu}\partial_{\nu}+\gamma^{\mu}_M\gamma^{\nu\dagger}_M\partial_{\mu}\partial_{\nu}\right)\\
 & =-\frac{1}{2}\left(\gamma^{\mu\dagger}_M\gamma^{\nu}_M+\gamma^{\mu}_M\gamma^{\nu\dagger}_M\right)\partial_{\mu}\partial_{\nu}\\
 & =-g_{E}^{\mu\nu}\partial_{\mu}\partial_{\nu}
\end{align}
where $g_E$ is the Euclidean metric.

\begin{rem}
  One could be tempted to substitute the Hilbert adjoint
    $D^\dagger$ with the Krein adjoint $D^+$ in \eqref{eq:66}, but it  is well known that this is problematic, for $D^+D$ is an
  hyperbolic operator, whereas the heat kernel technique used for the
  asymptotic expansion of $S^B$ are well defined only for elliptic
  operators.
\end{rem}
\section{Conclusions and Outlook}
The twist of the spectral triple corresponding to the Standard Model makes the
  Lorentzian signature naturally emerge as a twisted-inner
  product. The spectral
  action can be modified accordingly in order to make sense either of a
  selfadjoint Dirac operator, or of a $\rho$-hermitian one.  The
  first choice permits one to maintain the usual  definition for the 
  twisted spectral triple, but restricts the gauge group to those
  unitaries whose twisted adjoint action preserves
  selfadjointness (see \cite{Landi:2017aa}). On the other hand, choosing a $\rho$-hermitian Dirac
  operator does not restrict the gauge group ($\rho$-hermicity is
  preserved by twisted fluctuations), but requires one to modify the
  definition of twisted spectral triples in order to accommodate an operator $D$ that is twisted selfadjoint rather than selfadjoint.

The modifications of the spectral action that we propose here do not yield the bosonic action in a Lorentzian signature, which is a
well-known and difficult problem. Nevertheless, twists
shed a new light on the problem, if one considers that the question is
not so much to obtain directly from a spectral formula the Einstein-Hilbert action in the Lorentzian signature, than to be able to implement Wick rotation in a coherent
way. Traditionally in quantum field theory, one begins in a given Lorentz signature, 
Wick rotates to perform some calculation, then Wick rotates
back to obtain physical predictions. So far in
noncommutative geometry, one starts with a bosonic action in Euclidean signature,
expands with heat kernel techniques and then
Wick rotates. The results of this paper
suggest to start with a Lorentzian signature, for
which a twist is adapted (for twisted fluctuations preserve
Krein hermicity,  whereas usual fluctuations do not). The spectral action \eqref{eq:66} then yields the
Einstein-Hilbert action in Euclidean signature, and physical
predictions are obtained by Wick rotating back to the Lorentzian model
one has started with. The added value of the twist is thus to prescribe a geometry  upon which to ``Wick rotate back'' \cite{DAndrea:2016aa}.

\subsection*{Acknowledgement}
FL acknowledges the support of the INFN Iniziativa Specifica GeoSymQFT and Spanish
MINECO under project MDM-2014-0369 of ICCUB (Unidad de Excelencia `Maria de Maeztu').  This work has been possible with the help of a Short Term Scientific Mission of SF, for this we thank the {\sl COST} action \emph{QSPACE}.

\appendix
\section{Appendix: twisted  gauge fluctuation}
Given a twisted spectral $(\A, \HH, D), \rho$, we check that the
twisted adjoint action \eqref{eq:6} of the unitaries of $\A$ yields
the twisted gauge transformation \eqref{eq:7}. In agreement with
\eqref{eq:10} and \eqref{eq:34}), for $u$ a unitary of $\A$, we write
\begin{equation}
U:=\text{Ad}(u)= uJuJ^{-1},\quad
\rho(U)=\text{Ad}(\rho(u))=\rho(u)J\rho(u) J^{-1}.
\label{eq:87}
\end{equation}
Assuming $\rho$ is a $^*$-automorphism (which is the case of the flip
in the Standard Model), one first checks that
\begin{eqnarray}
\rho(U)DU^{*} & = & \rho(u)J\rho(u)J^{*}DJu^{*}J^{*}u^{*}\nonumber \\
 & = & \epsilon'\rho(u)J\rho(u)Du^{*}J^{*}u^{*}\,\,\,\,\,\,\,\mbox{using }DJ=\epsilon'JD\nonumber \\
 & = &\epsilon' \rho(u)J\rho(u)\left(\rho(u^{*})D+[D,u^{*}]_\rho\right)J^{*}u^{*}\nonumber \\
 & = & \epsilon' \rho(u)JDJ^{*}u^{*}+\epsilon'\rho(u)J\rho(u)\left[D,u^{*}\right]_{\rho}J^{*}u^{*}\nonumber \\
 & = & \rho(u)Du^{*}+ \epsilon'JJ^{*}\rho(u)J\rho(u)\left[D,u^{*}\right]_{\rho}J^{*}u^{*}\nonumber \\
 & = & \rho(u)Du^{*}+\epsilon'J\rho(u)J^{*}\rho(u)J\left[D,u^{*}\right]_{\rho}J^{*}u^{*}\,\,\,\mbox{using}\left[J^{*}\rho(u)J,\rho(u)\right]=0\nonumber \\
 & = & \rho(u)Du^{*}+\epsilon'J\rho(u)\left[D,u^{*}\right]_{\rho}J^{*}uJJ^{*}u^{*}\,\,\,\mbox{using}\left[\left[D,u^{*}\right]_{\rho},J^{*}\rho(u)J\right]_{\rho}=0\nonumber \\
 & = & \rho(u)Du^{*}+\epsilon'J\rho(u)\left[D,u^{*}\right]_{\rho}J^{*}\nonumber \\
 & = & \rho(u)\left(\rho(u^{*})D+[D,u^{*}]_\rho\right)+ \epsilon'J\rho(u)\left[D,u^{*}\right]_{\rho}J^{*}\nonumber \\
 & = & D+\rho(u)\left[D,u^{*}\right]_{\rho}+\epsilon'J\rho(u)\left[D,u^{*}\right]_{\rho}J^{*}.\label{eq:proof0}
\end{eqnarray}
Then, noticing that by the order zero and the twisted first order
condition one has
\begin{equation}
  \label{eq:88}
  [A,J\rho(a)J^*]_\rho =0
\end{equation}
for any twisted $1$-form $A= a^i[D, b_i]_\rho$, one has
\begin{eqnarray}
\rho(U)AU^{*} & = & \rho(u)J\rho(u)J^{*}AJu^{*}J^{*}u^{*}\nonumber \\
 & = & \rho(u)J\rho(u)J^{*}J\rho(u^{*})J^{*}Au^{*}\,\,\mbox{using}\left[A,J\rho(u^{*})J^{*}\right]_{\rho}=0\nonumber \\
 & = & \rho(u)Au^{*}.\label{eq:proof1}
\end{eqnarray}
As well,
\begin{eqnarray}
\rho(U)JAJ^{*}U^{*} & = & \rho(u)J\rho(u)J^{*}JAJ^{*}Ju^{*}J^{*}u^{*}\nonumber \\
 & = & \rho(u)J\rho(u)Au^{*}J^{*}u^{*}JJ^{*}\nonumber \\
 & = & \rho(u)J\rho(u)AJ^{*}u^{*}Ju^{*}J^{*}\,\,\mbox{using}\left[J^{*}u^{*}J,u^{*}\right]=0\,\nonumber \\
 & = & \rho(u)J\rho(u)J^{*}\rho(u^{*})JAu^{*}J^{*}\,\,\mbox{using}\left[A,J^{*}u^{*}J\right]_{\rho}=0\nonumber \\
 & = & J\rho(u)J^{*}\rho(u)\rho(u^{*})JAu^{*}J^{*}\,\,\mbox{using}\left[J\rho(u)J^{*},\rho(u)\right]=0\nonumber \\
 & = & J\rho(u)Au^{*}J^{*}.\label{eq:proof2}
\end{eqnarray}
Therefore, collecting \eqref{eq:proof0}, \eqref{eq:proof1} and
\eqref{eq:proof2}, one finds that a twisted covariant operator $D_A= D+ A + \epsilon' JAJ^*$
is mapped under a twisted gauge trasformation to 
  \begin{equation}
\rho(U)D_{A} U^*  = 
D+\rho(u)Au^{*}+\rho(u)\left[D,u^{*}\right]_{\rho}+\epsilon'J\left(\rho(u)Au^{*}+\rho(u)\left[D,u^{*}\right]_{\rho}\right)J^{*}\label{eq:5}
\end{equation}
which is nothing but $D + A^u + JA^uJ^*$ for $A^u$ the twisted gauge
transform \eqref{eq:7}.

 \end{document}